%% file: manuscript_PRA_rep1.tex
\documentclass[pra,twocolumn,showpacs]{revtex4-2}

\usepackage{amsmath}
\usepackage{latexsym}
\usepackage{amssymb}
\usepackage{graphicx}
\usepackage[colorlinks=true, citecolor=blue, urlcolor=blue]{hyperref}
\usepackage{float}
\usepackage{amsfonts}
\usepackage{textcomp}
\usepackage{mathpazo}
\usepackage{comment}
\usepackage{xr}

\usepackage{bbm}

\usepackage{xcolor}
\definecolor{myurlcolor}{rgb}{0,0,0.4}
\definecolor{mycitecolor}{rgb}{0,0.5,0}
\definecolor{myrefcolor}{rgb}{0.5,0,0}
\usepackage{hyperref}
\hypersetup{colorlinks,
linkcolor=myrefcolor,
citecolor=mycitecolor,
urlcolor=myurlcolor}

\usepackage[draft]{fixme}
\usepackage{amsmath,bbm}
\usepackage{graphicx}
\usepackage{amsfonts}
\usepackage{amssymb}
\usepackage{amsmath, amssymb, amsthm,verbatim,graphicx,bbm}
\usepackage{mathrsfs}
\usepackage{color,xcolor,longtable}

\usepackage{xr}
\makeatletter
\newcommand*{\addFileDependency}[1]{
  \typeout{(#1)}
  \@addtofilelist{#1}
  \IfFileExists{#1}{}{\typeout{No file #1.}}
}
\makeatother

\newcommand*{\myexternaldocument}[1]{
    \externaldocument{#1}
    \addFileDependency{#1.tex}
    \addFileDependency{#1.aux}
}

\myexternaldocument{manuscript_PRL}

\newcommand{\beq}[0]{\begin{equation}}
\newcommand{\eeq}[0]{\end{equation}}

\newcommand{\one}{\leavevmode\hbox{\small1\normalsize\kern-.33em1}}

\def\be{\begin{equation}}
\def\ee{\end{equation}}
\def\ben{\begin{eqnarray}}
\def\een{\end{eqnarray}}
\def\eea{\end{array}}
\def\bea{\begin{array}}

\newcommand{\Tr}[1]{\mathrm{Tr}#1}
\newcommand{\bei}{\begin{itemize}}
\newcommand{\eei}{\end{itemize}}
\newcommand{\ket}[1]{|#1\rangle}
\newcommand{\bra}[1]{\langle#1|}

\newcommand{\proj}[1]{\ket{#1}\!\!\bra{#1}}

\newcommand{\I}{\mathbbm{1}}

\renewcommand{\emph}[1]{\textbf{#1}}

\makeatletter
\newtheorem*{rep@theorem}{\rep@title}
\newcommand{\newreptheorem}[2]{%
\newenvironment{rep#1}[1]{%
 \def\rep@title{#2 \ref{##1}}%
 \begin{rep@theorem}}%
 {\end{rep@theorem}}}
\makeatother

\theoremstyle{plain}
\newtheorem*{thm*}{Theorem}
\newreptheorem{thm}{Theorem}

\newtheorem{fakt}{Fact}

\theoremstyle{definition}

\theoremstyle{remark}

\usepackage[T1]{fontenc}

\begin{document}

\title{Witnessing network steerability of every bipartite entangled state without inputs}
\author{Shubhayan Sarkar}
\email{shubhayan.sarkar@ug.edu.pl}
\affiliation{Laboratoire d’Information Quantique, Université libre de Bruxelles (ULB), Av. F. D. Roosevelt 50, 1050 Bruxelles, Belgium}
\affiliation{Institute of Informatics, Faculty of Mathematics, Physics and Informatics,
University of Gdansk, Wita Stwosza 57, 80-308 Gdansk, Poland}

\begin{abstract}	
Quantum steering is an asymmetric form of quantum nonlocality where one can detect whether a measurement on one system can steer or change another distant system. It is well-known that there are quantum states that are entangled but unsteerable in the standard quantum steering scenario. Consequently, a long-standing open problem in this regard is whether the steerability of every entangled state can be activated in some way. In this work, we consider quantum networks and focus on the swap-steering scenario without inputs and find linear witnesses of network steerability corresponding to any negative partial transpose (NPT) bipartite state and a large class of bipartite states that violate the 
computable cross-norm (CCN) criterion. Interestingly, one of the inequalities shows an arbitrarily large gap between network steerable and unsteerable models. This is, in fact, the first instance where such an unbounded gap could be established between any form of network local and nonlocal correlations, and notably, it arises in the case of no inputs from either party. Furthermore, by considering that the trusted party can perform tomography of the incoming subsystems, we construct linear inequalities to witness swap-steerability of every bipartite entangled state.  Consequently, for every bipartite entangled state one can now observe a form of quantum steering. 
\end{abstract}


\maketitle

{\section{Introduction}} Entangled quantum states are a class of quantum states that can not expressed as a mixture of product states. Nonlocality, on the other hand, points to the fact that there exist correlations between space-like separated systems that cannot be explained by classical physics. An asymmetric form of quantum nonlocality is known as quantum steering which essentially captures the fact that two distant quantum systems can influence each other's state \cite{Schrod}. The standard quantum steering scenario consists of two parties among which one is trusted in the sense that the measurements of this party are precisely known \cite{Wiseman}. From an application perspective, quantum steering serves as a resource for various quantum information tasks such as key distribution \cite{bran,Tomamichel_2013,masini}, randomness certification \cite{paul1,sarkar11}, self-testing of quantum states and measurements \cite{Supic, sarkar6,sarkar11,sarkar12}, and more.

To observe quantum steering, one requires entangled quantum systems. Consequently, quantum steering implies the presence of entanglement. It is well-established now that there are entangled states that can not demonstrate quantum steering in the standard scenario \cite{Uola_2020}. For instance, Werner state beyond a certain parameter value is proven to be unsteerable even if it is entangled \cite{Bowles1}.
A major problem in this regard concerns whether the steerability of every entangled state can be activated in some manner. It was demonstrated in \cite{act1} that certain unsteerable entangled states can become steerable when multiple copies of the state are considered. Consequently, this indicates that if one considers quantum networks which employ multiple sources, one might be able to demonstrate the activation of quantum steering for every entangled state.

Recently, a form of quantum steering in networks \cite{netstee}, termed swap-steering, was introduced in \cite{sarkar15}. The scenario consisted of two spatially separated parties, each performing a fixed measurement, with one of them being trusted. This is the minimal scenario to observe any form of network nonlocality. In this scenario, we first construct a family of witnesses to observe swap-steerability, the minimal form of network steerability, of every bipartite state with a negative partial transpose (NPT) \cite{Peres,Horodecki_1996}. Then, we construct witnesses of swap-steerability for a large class of states violating the computable cross-norm (CCN) criterion \cite{Yu_2005} which includes some bound entangled states. Finally, allowing the trusted party to perform tomography of the received subsystem, we show that if one can construct an entanglement witness of a particular bipartite entangled state, then one can straightforwardly construct a witness to observe swap-steerability. Consequently, every bipartite entangled state is network steerable.

Interestingly, one of the proposed inequalities allows us to observe an unbounded gap between network quantum steering and network unsteerable models. This is the first example of an unbounded gap between any form of network local and non-local correlations. Remarkably, this happens in the case without inputs for both parties, which is again the first of its kind. From a foundational point of view, this means that there is no way network unsteerable models can even approximate network steerable correlations. Moreover, in realistic settings, imperfections such as noise, finite statistics, and detector inefficiencies inevitably reduce the observed violation. A large gap provides a safety margin that ensures the measured correlations remain clearly incompatible with any network steerable and unsteerable models despite these imperfections. This improves statistical significance, reduces the precision and data requirements of the experiment, and enhances robustness against loopholes such as limited detection efficiency. Considering higher-dimensional systems allows to observe swap-steering in a very noisy setup, thus making it more experimental friendly. 

Before proceeding with the result, let us first introduce the relevant concepts required for the manuscript. 

{\section{Preliminaries}}
Let us briefly describe the swap-steering scenario introduced in \cite{sarkar15}. This setup involves two parties, Alice and Bob, who are situated in separate laboratories. Each of them receives two subsystems from two distinct and statistically independent sources, $S_1$ and $S_2$. They then perform a single measurement on their respective subsystems, with the outcomes labelled as $a, b$ for Alice and Bob, respectively [see Fig. \ref{fig1}]. In this scenario, Alice is considered trustworthy, meaning her measurement apparatus is well-calibrated and known. Alice and Bob repeat this experiment multiple times to generate the joint probability distribution, represented by $\vec{p} = {p(a, b)}$, where $p(a, b)$ denotes the probability of Alice and Bob obtaining the outcomes $a$ and $b$, respectively. Let us note here that in general the sources $S_1,S_2$ can be classically correlated as shown in \cite{sarkar15}. However, for simplicity, we assume that both sources are independent for our analysis all of which follow straightway to the scenario with classically correlated sources as all the proposed witnesses are linear over $\vec{p}$. 

The probabilities $p(a,b)$ can be computed in quantum theory as
\begin{eqnarray}
p(a,b)=\Tr\left[(M^a\otimes N^b)\rho_{A_1B_1}\otimes\rho_{A_2B_2}\right]
\end{eqnarray}
where $M^a,N^b$ denote the measurement elements of Alice and Bob which are positive and $\sum_aM^a=\sum_bN^b=1$. 
If the correlations are not swap-steerable, then one can always find a separable outcome-independent hidden state (SOHS) model to explain the observed statistics $\vec{p}$ \cite{sarkar15} given by
\begin{eqnarray}\label{SOHS2}
  p(a,b)= \sum_{\lambda_1,\lambda_2} p(\lambda_1)p(\lambda_2)p(a|\ \rho_{\lambda_1}\otimes\rho_{\lambda_2})p(b|\lambda_1,\lambda_2)\ \ 
\end{eqnarray}
for all $a,b$. To observe swap-steering one can construct linear inequalities of the form
\begin{eqnarray}
    W=\sum_{a,b}c_{a,b}p(a,b) \leq\beta_{SOHS}
\end{eqnarray}
where $c_{a,b}$ are scalars
and $\beta_{SOHS}$ is the maximal value attainable using SOHS model.

In this work, we aim to find linear inequalities in the above-described scenario, that can be used to witness swap-steerability for any bipartite entangled states. For this purpose, let us describe two necessary criteria for any bipartite state to be entangled. The first one is the well-known positive partial transpose (PPT) criterion \cite{Peres,Horodecki_1996} which can be simply stated as any separable state must have a positive partial transpose. Consequently, quantum states with negative partial transpose (NPT) are guaranteed to be entangled. This further allows one to construct entanglement witnesses corresponding to each of the NPT states as
\begin{eqnarray}\label{npt}
W_{NPT}=\Tr(\proj{\eta}^{T_A}_{AB}\ \rho_{AB})
\end{eqnarray}
where $T_A$ denotes the partial transpose over subsystem $A$ and $\ket{\eta}$ is the eigenvector of $\rho_{AB}^{T_A}$ with a negative eigenvalue. Consequently, $W_{NPT}\geq 0$ for separable states and $W_{NPT}< 0$ for the corresponding entangled state. 

The second criterion is known as the computable cross-norm (CCN) criterion \cite{Yu_2005}. To describe the CCN criterion, we express any bipartite density matrix using the Schmidt decomposition in the operator space as
\begin{eqnarray}
    \rho_{AB}=\sum_{i=0}^{d^2-1}\lambda_iF_{i,A}\otimes G_{i,B}
\end{eqnarray}
where $\lambda_i\geq 0$ and $\{F_{i,A}\},\{G_{i,B}\}$ are orthonormal bases of the operator space of $\mathcal{H}_A, \mathcal{H}_B$ with $\min\{\mathrm{dim}(\mathcal{H}_A),\mathrm{dim}(\mathcal{H}_B)\}=d$, that is, $\Tr(F_{i}F_j)=\Tr(G_{i}G_j)=\delta_{i,j}$ such that $F_i,G_i$ are hermitian. An example of such a basis for operators acting on $\mathbb{C}^d$ is given by $\{\mathcal{J}_m,\mathcal{J}_{m,n}^{\pm}\}$ with $\mathcal{J}_m=\proj{m}$
\begin{eqnarray}\label{maxentccn1}
\mathcal{J}_{m,n}^{+}=\frac{\ket{m}\!\bra{n}+\ket{n}\!\bra{m}}{\sqrt{2}},\quad \mathcal{J}_{m,n}^{-}=\frac{\ket{m}\!\bra{n}-\ket{n}\!\bra{m}}{i\sqrt{2}}
\end{eqnarray}
for $m<n$ such that $m,n=0,1,\ldots,d-1$. The maximally entangled state of two-qudits $\ket{\phi^+_d}=1/\sqrt{d}(\sum_i\ket{ii})$ can be represented using this bases as
\begin{eqnarray}\label{maxentccn}
    \proj{\phi^+_d}=\frac{1}{d}\sum_{m=0}^{d-1} \mathcal{J}_m\otimes \mathcal{J}_m+\frac{1}{d}\sum_{\substack{m,n=0\\m<n}}^{d-1}\mathcal{J}_{m,n}^{\pm}\otimes (\mathcal{J}_{m,n}^{\pm})^T.\ \ \ 
\end{eqnarray}
Now, the CCN criterion states that for any separable state $\sum_{i=0}^{d^2-1}\lambda_i\leq1.$ Consequently, quantum states that violate this criterion are entangled. Both of the above-described criteria are necessary for quantum states to be entangled. For other necessary conditions, refer to \cite{G_hne_2009}. Let us now construct witnesses that can detect swap-steerability of entangled states that violate the above conditions.

\begin{figure}[t]
\includegraphics[width=\linewidth]{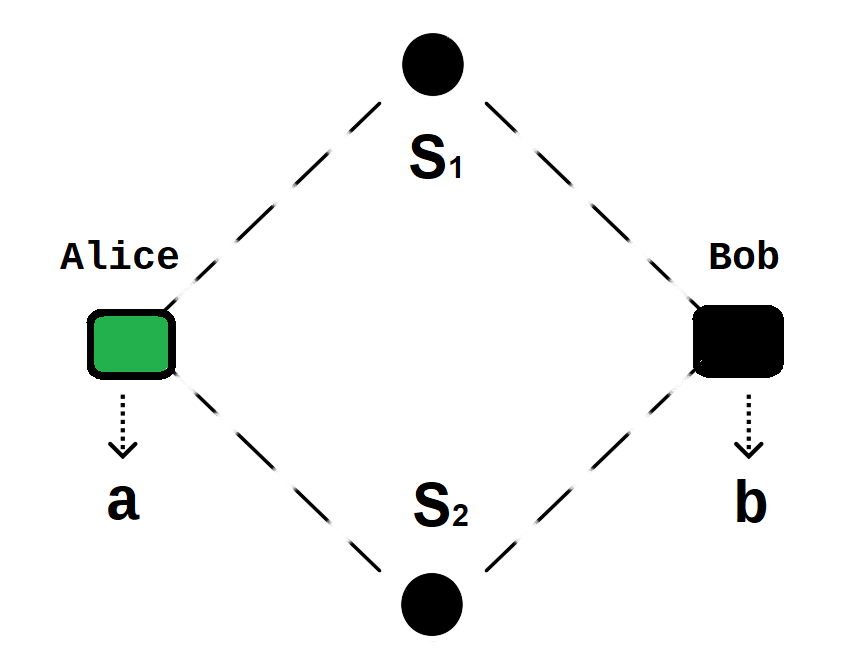}
    \caption{Swap-steering scenario. Alice and Bob, situated apart with no communication between them, each receives two subsystems from sources $S_1$ and $S_2$. They independently perform a single measurement on these subsystems with Alice deemed trustworthy. }
    \label{fig1}
\end{figure}

{\section{NPT states}} Consider again the scenario depicted in Fig. \ref{fig1} with Bob performing a two-outcome measurement $b=0,1$ and Alice being trusted performs the  $d^2-$outcome measurement $U_{A_1}^{\dagger}\mathcal{A}_{NPT}U_{A_1}$ where $U_{A_1}$ is a unitary operation (described below) and $\mathcal{A}_{NPT}=\{\proj{\delta_{m}},\proj{\delta^{\pm}_{m,n}}\}$ with $\ket{\delta_{m}}=\ket{m}_{A_1}\ket{m}_{A_2}$ and 
\begin{eqnarray}\label{ANPT}
\ket{\delta_{m,n}^{\pm}}=\frac{\ket{m}_{A_1}\ket{n}_{A_2}\pm\ket{n}_{A_1}\ket{m}_{A_2}}{\sqrt{2}}
\end{eqnarray}
such that $m,n=0,1,\ldots,d-1$ with $m<n$. For simplicity, Alice's outcome will be denoted here as $a=m$ and $a=(\pm,m,n)$. 

Suppose now that we want to witness swap-steerability of an NPT state $\rho_{A_1B_1}$ with $\min\{\mathrm{dim}(\mathcal{H}_{A_1}),\mathrm{dim}(\mathcal{H}_{B_1})\}=d$. As discussed earlier, we consider the eigenstate $\proj{\eta}$ associated with a negative eigenvalue of $\rho_{A_1B_1}^{T_{A_1}}.$ As $\ket{\eta}\in \mathbb{C}^d\otimes\mathbb{C}^d$ is a bipartite state, it can be expressed using Schmidt decomposition as $\ket{\eta}=\sum_{i=0}^{d-1}\alpha_i\ket{e_i}_{A_1}\ket{f_i}_{B_1}$. As $\{\ket{e_i}\},\{f_i\}$ are orthonormal bases, we have   $\ket{\tilde{\eta}}=U_{A_1}\otimes V_{B_1}\ket{\eta}=\sum_{i=0}^{d-1}\alpha_i\ket{i}_{A_1}\ket{i}_{B_1}$ where $U\ket{e_i}=\ket{i}$ and $V\ket{f_i}=\ket{i}$. 

Now, the corresponding swap-steering inequality is given by $\mathcal{S}_{\rho:NPT}\leq\beta_{SOHS}$ where
\begin{eqnarray}\label{Witnpt}
    \mathcal{S}_{\rho:NPT}=\sum_{\substack{m,n=0\\m<n}}^{d-1}\alpha_m\alpha_n\Big[p((-,m,n),0)-p((+,m,n),0)\Big]\nonumber\\-\sum_{m=0}^{d-1}\alpha_m^2p(m,0).\qquad\qquad\qquad
\end{eqnarray}
Let us recall that Alice's outcomes are denoted here as $a=m$ and $a\equiv(\pm,m,n)$ [see Eq. \eqref{ANPT}] and thus $p((\pm,m,n),0)$ signifies the probability of obtaining outcome $(\pm,m,n)$ by Alice and $0$ by Bob.
Consequently, we establish the following result for the above-suggested inequality.
\begin{fakt}
    Consider the swap-steering scenario depicted in Fig. \ref{fig1} and the functional $\mathcal{S}_{\rho:NPT}$ \eqref{Witnpt}. The maximal value attainable of $\mathcal{S}_{\rho:NPT}$ using an SOHS model is $\beta_{SOHS}=0$. 
\end{fakt}

The proof of the above fact can be found in Appendix A. Consider now that the source $S_1$ generates the state $\rho_{A_1B_1}$ and $S_2$ generates the maximally entangled state $\ket{\phi^+_d}_{A_2B_2}$ with Bob performing the measurement $\{M_0,\I-M_0\}$ where 
$M_0=V_{B_1}^{\dagger}\proj{\phi^+_d}_{B_1B_2}V_{B_1}$. It is straightforward to see using entanglement swapping that when Bob obtains outcome $0$, the post-measurement state at Alice's side is given by $\sigma_{A_1A_2}^0= V_{A_2}\rho_{A_1A_2}V_{A_2}^{\dagger}$ that occurs with probability $1/d^2$. Now, we observe from \eqref{Witnpt} that 
\begin{eqnarray}\label{Witnip2}
\mathcal{S}_{\rho:NPT}&=&\sum_{\substack{m,n=0\\m<n}}^{d-1}\frac{\alpha_m\alpha_n}{d^2}\Tr\Big[(\proj{\delta_{m,n}^-}-\proj{\delta^+_{m,n}})\tilde{\rho}\Big]\nonumber\\&-&\sum_{m=0}^{d-1}\frac{\alpha_m^2}{d^2}\Tr[\proj{\delta_m}\tilde{\rho}]
\end{eqnarray}
where $\tilde{\rho}_{A_1A_2}=U_{A_1}\otimes V_{A_2}\rho_{A_1A_2}U_{A_1}^{\dagger}\otimes V_{A_2}^{\dagger}.$
Let us now observe that [see Appendix for details]
\begin{eqnarray}
    \proj{\tilde{\eta}}^{T_A}&=&\sum_{\substack{m,n=0\\m<n}}^{d-1}\alpha_m\alpha_n(\proj{\delta_{m,n}^+}-\proj{\delta^-_{m,n}})\nonumber\\&+&\sum_{m=0}^{d-1}\alpha_m^2\proj{\delta_m}
\end{eqnarray}
which allows us to conclude from \eqref{Witnip2} that
\begin{eqnarray}
    \mathcal{S}_{\rho:NPT}=-\frac{1}{d^2}\Tr(\proj{\eta}^{T_A}\rho)>0.
\end{eqnarray}
Consequently, any NPT state is swap-steerable. 

{\section{States violating CCN criterion---}} Suppose now that we want to witness swap-steerability of a state $\rho_{A_1B_1}$ with $\min\{\mathrm{dim}(\mathcal{H}_A),\mathrm{dim}(\mathcal{H}_B)\}=d$ that violates the CCN criterion, that is, we can express the state $\rho_{A_1B_1}$ as
\begin{equation}
    \rho_{A_1B_1}=\sum_{m=0}^{d-1} \lambda_mF_m\otimes G_m+\sum_{\substack{m,n=0\\m<n}}^{d-1}\lambda_{\pm,m,n}F_{m,n}^{\pm}\otimes G_{m,n}^{\pm}
\end{equation}
with $\lambda_m,\lambda_{\pm,m,n}\geq0$ and $\sum_{m=0}^{d-1}\lambda_m+\sum_{\substack{m,n=0\\m<n}}^{d-1}\lambda_{\pm,m,n}>1$ and $\{F_m,F_{m,n}^{\pm}\},\{G_m,G_{m,n}^{\pm}\}$ are set of orthonormal bases on $\mathcal{H}_A,\mathcal{H}_B$ respectively. Let us restrict to states that can also be expressed up to local unitaries as
\begin{eqnarray}\label{stateccn}
    U'_{A_1}\otimes V'_{B_1} &\rho_{A_1B_1}& U_{A_1}^{'\dagger}\otimes V_{B_1}^{'\dagger}=\sum_{m=0}^{d-1} \lambda_m\mathcal{J}_m\otimes \mathcal{J}_m\nonumber\\&+&\sum_{\substack{m,n=0\\m<n}}^{d-1}\lambda_{\pm,m,n}\mathcal{J}_{m,n}^{\pm}\otimes (\mathcal{J}_{m,n}^{\pm})^T
\end{eqnarray}
where $U'_{A_1}F_{m,n}U_{A_1}^{\dagger}=\mathcal{J}_{m,n}$, $V'_{A_1}G_{m,n}V_{A_1}^{\dagger}=\mathcal{J}_{m,n}^T$,
$U'_{A_1}F_{m}U_{A_1}^{\dagger}=\mathcal{J}_{m}$, and $V'_{A_1}G_{m}V_{A_1}^{\dagger}=\mathcal{J}_{m}$ for all $m,n$ and $\mathcal{J}_{m},\mathcal{J}^{\pm}_{m,n}$ are given in \eqref{maxentccn1}. 

Consider again the scenario depicted in Fig. \ref{fig1} with Alice and Bob performing $d^2$-outcome measurement $a,b=l_1l_2$ with $l_1,l_2=0,1,\ldots,d-1$. Alice being trusted performs the Bell-state measurement up to local unitary given by $U_{A_1}^{'\dagger}\mathcal{A}_{BM}U'_{A_1}$ where $U'_{A_1}$ is a unitary operation (described above) and $\mathcal{A}_{BM}=\{\proj{\phi_{d,l_1l_2}^+}\}$ with 
\begin{eqnarray}\label{maxentbasis}
\ket{\phi_{d,l_1l_2}^+}=\left(X^{l_2}_dZ^{l_1}_d\otimes\I\right)\frac{1}{\sqrt{d}}\sum_i\ket{i}\ket{i}.
\end{eqnarray}
such that $Z_d=\sum_{i=0}^{d-1}\omega_d^i\proj{i},X_d=\sum_{i=0}^{d-1}\ket{i+1}\!\bra{i}$ with $\omega_d=e^{2\pi i/d}$. Now, the corresponding swap-steering inequality is given by $\mathcal{S}_{\rho:CCN}\leq\beta_{SOHS}$ where
\begin{eqnarray}\label{Witccn}
    \mathcal{S}_{\rho:CCN}&=&\sum_{l_1,l_2=0}^{d-1}p(l_1l_2,l_1l_2).
\end{eqnarray}
Let us observe that the above inequality \eqref{Witccn} is the $d^2-$outcome generalisation of the swap-steering inequality of \cite{sarkar15}.
Consequently, we establish the following result for the above-suggested inequality.
\begin{fakt}
    Consider the swap-steering scenario depicted in Fig. \ref{fig1} and the functional $\mathcal{S}_{\rho:CCN}$ \eqref{Witccn}. The maximal value attainable of $\mathcal{S}_{\rho:CCN}$ using an SOHS model is $\beta_{SOHS}=\frac{1}{d}$. 
\end{fakt}

The proof of the above fact can be found in Appendix A.
Consider now that the source $S_1$ generates the state $\rho_{A_1B_1}$ and $S_2$ generates the maximally entangled state $\ket{\phi^+_d}_{A_2B_2}$ with Bob performing the Bell state measurement up to local unitary $V_{B_1}^{'\dagger}\mathcal{A}_{BM}^{*}V'_{B_1}$ where $\mathcal{A}_{BM}$ is given in \eqref{maxentbasis} and $*$ denotes complex conjugation. It is straightforward to see using entanglement swapping that when Bob obtains outcome $l_1l_2$, the post-measurement state at Alice's side is given by $\sigma_{A_1A_2}^{l_1l_2}= [\I_{A_1}\otimes(X^{l_2}_dZ^{l_1}_d)^TV'_{A_2}]\rho_{A_1A_2}[\I_{A_1}\otimes V^{'\dagger}(X^{l_2}_dZ^{l_1}_d)^*_{A_2}]$ that occurs with probability $1/d^2$. Now, we observe from \eqref{Witccn} that 
\begin{eqnarray}\label{Witnip21}
\mathcal{S}_{\rho:CCN}&=&\frac{1}{d^2}\sum_{l_1,l_2=0}^{d-1}\Tr\left(\sigma_{A_1A_2}^{l_1l_2}U_{A_1}^{'\dagger}\proj{\phi_{d,l_1l_2}^+}U'_{A_1}\right)\nonumber\\
&=&\Tr(\tilde{\rho}_{A_1A_2}\proj{\phi^+_d})
\end{eqnarray}
where $\tilde{\rho}_{A_1A_2}=U_{A_1}\otimes V_{A_2}\rho_{A_1A_2}U_{A_1}^{\dagger}\otimes V_{A_2}^{\dagger}$ and to obtain the second line we used the fact in \eqref{maxentbasis} that $R\otimes\I\ket{\phi^+_d}=\I\otimes R^T\ket{\phi^+_d}$. Now, using the Schmidt decomposition of $\tilde{\rho}_{A_1A_2}$ from \eqref{stateccn} and $\proj{\phi^+_d}$ from \eqref{maxentccn} and then using the fact that all the basis elements in the representation are orthonormal, we obtain that
\begin{eqnarray}
    \mathcal{S}_{\rho:CCN}=\frac{1}{d}\sum_{m=0}^{d-1}{\lambda_m}+\frac{1}{d}\sum_{\substack{m,n=0\\m<n}}^{d-1}\lambda_{\pm,m,n}> \frac{1}{d}.
\end{eqnarray}
Let us now observe an interesting property of the swap-steering inequality $\mathcal{S}_{\rho:CCN}$ \eqref{Witccn}.

{\it{Unbounded gap to witness swap-steerability---}} Consider again the expression \eqref{Witnip21} and observe that if $\tilde{\rho}=\proj{\phi^+_d}$, then $\mathcal{S}_{\rho:CCN}=1$. Moreover, for every dimension $d$, $0\leq\mathcal{S}_{\rho:CCN}\leq1$. Thus, the ratio of the maximal quantum value to the maximal value that one can obtain via local or SOHS model scales up as $d$. This is particularly important for experiments as noise increases with dimensions but the gap between the local and quantum values also increases with $d$. From a fundamental perspective, this further points to the fact that there exist operational tasks without inputs in which quantum strategies scalably outperform local ones. 

{\section{Witness for any bipartite entangled state}} Let us now generalise the above scenario for witnessing swap-steering for arbitrary bipartite entangled state. For this purpose, let us recall from the Hahn-Banach separation theorem that there exists a witness $W_{\tilde{\rho}}$ for every entangled state $\rho$ such that $\Tr(W_{\tilde{\rho}}\sigma)\geq0$ where $\sigma$ denotes any separable state and $\Tr(W_{\tilde{\rho}}\tilde{\rho})<0$. For instance, Eq. \eqref{npt} is a witness for every NPT state. Consider now that $\tilde{\rho}$ acting on $\mathbb{C}^{d'}\otimes\mathbb{C}^{d''}$ which can be embedded in $\mathbb{C}^{d}\otimes\mathbb{C}^{d}$ where $d=\max\{d,d'\}$. Consequently, $W_{\tilde{\rho}}$ can be expressed in the Heisenberg-Weyl (HW) basis $\{\omega_{d^2}^{ij(d^2-1)/2}X^i_{d^2}Z^j_{d^2}\}$ for $i,j=0,1,\ldots,d^2-1$ as $W_{\tilde{\rho}}=\sum_{i,j=0}^{d^2-1}\lambda_{i,j}\omega_{d^2}^{ij(d^2-1)/2}X^i_{d^2}Z^j_{d^2}$. Notice that we add factors in front of the standard HW basis such that each of them is a proper observable with eigenvalues as powers of $\omega$. As $W_{\tilde{\rho}}$ is hermitian, we further obtain that $\omega_{d^2}^{ij}\lambda_{i,j}^{*}=\lambda_{j,i}$.

Let us now consider the weakened version of the swap-steering scenario depicted in Fig. \ref{fig1} such that the trusted party can perform $d^2+1$ number of $d^2-$outcome measurements rather than just a single one written in the observable form as $\{A_{01},A_{10},A_{11},\ldots,A_{1d^2-1}\}$ where $A_{01}=Z_{d^2}$ and $A_{1k}=\omega_{d^2}^{k(d^2-1)/2}X_{d^2}Z^k_{d^2}\ (k=0,\ldots,d^2-1)$. In fact, these measurements can be used to construct the full tomographically complete set of measurements on Alice's side \cite{sarkar12}.  However, the untrusted party performs a single two-outcome measurement. Consequently, in this scenario, we propose the following inequality:
\begin{equation}\label{univwit}
\mathcal{S}_{\tilde{\rho}}=c_{0,0}p_B(0)+\sum_{a=0}^{d^2-1}c_{a,01}p(a,0|01)+\sum_{a,k=0}^{d^2-1}c_{a,1k}p(a,0|1k)
\end{equation}
where $p_B(0)$ is the local probability of Bob's $0-$outcome and $c_{0,0}=-\lambda_{0,0},c_{a,01}=-\sum_{j=1}^{d^2-1}\lambda_{0,j}\omega^{aj}_{d^2}$ and $ c_{a,1k}=-\sum_{j=1}^{d^2-1}\lambda_{k,kj\oplus d^2}\omega^{aj}_{d^2}$ where $\{\lambda_{i,j}\}$ are coefficients of $W_{\tilde{\rho}}$ as described above and thus $c_{i,j}'s$ are real. Let us now compute the local bound $\beta_{SOHS}$ of $\mathcal{S}_{\tilde{\rho}}$ \eqref{univwit}.

\begin{fakt}\label{fact3}
    Consider the swap-steering scenario depicted in Fig. \ref{fig1} and the functional $\mathcal{S}_{\tilde{\rho}}$ \eqref{univwit}. The maximal value attainable of $\mathcal{S}_{\tilde{\rho}}$ using an SOHS model is $\beta_{SOHS}=0$. 
\end{fakt}
The proof of the above fact can be found in Appendix A. Consider now that the source $S_1$ generates the state $\tilde{\rho}_{A_1B_1}$ and $S_2$ generates the maximally entangled state $\ket{\phi^+_d}_{A_2B_2}$ with Bob performing the measurement $\{M_0,\I-M_0\}$ where 
$M_0=\proj{\phi^+_d}_{B_1B_2}$. It is straightforward to see using entanglement swapping that when Bob obtains outcome $0$, the post-measurement state at Alice's side is given by $\sigma_{A_1A_2}^0= \tilde{\rho}_{A_1A_2}$ that occurs with probability $1/d^2$. Now, we observe from \eqref{univwit} that can be simplified to [see proof of Fact \ref{fact3}] 
\begin{eqnarray}
    \mathcal{S}_{\tilde{\rho}}=-\frac{1}{d^2}\Tr(W_{\tilde{\rho}}\tilde{\rho})>0.
\end{eqnarray}
Consequently, any bipartite entangled state is swap-steerable. 

{\section{Discussions}} In the above work, we first constructed a family of witnesses to detect swap-steerability for every NPT bipartite state. Additionally, we constructed swap-steerability witnesses for a large class of states that violate the CCN criterion. Remarkably, one of our proposed inequalities demonstrates an unbounded gap between network steerable and unsteerable correlations. This is, in fact, the first instance where such an unbounded gap could be established between any form of network local and nonlocal correlations, that too  without inputs from either party.  Moreover, by allowing the trusted party to perform tomography on the received subsystem, we showed that an entanglement witness for a particular bipartite entangled state can be directly adapted to construct a witness for swap-steerability. Notice that to observe swap-steering, that is, to violate the above proposed witnesses, we consider that one of the states generated by the sources is maximally entangled. To ensure that the observed steerability is due to both states and not just due to the chosen maximally entangled state, we show in Appendix B that if either state is separable, then the proposed inequalities can not be violated. Consequently, entanglement in both sources is necessary to observe swap-steering [see Ref. \cite{sarkar15} for a more general condition]. It should be noted that from an experimental point of view, the proposed witnesses are easy to implement. For a fixed local dimension, the experimental setup requires performing tomography on one side and post-selecting a single outcome of a Bell measurement, both of which can be realized using linear optics. This is enough to observe swap steerability of any bipartite quantum state.

Let us further interpret the above activation of quantum steering from a phenomenological point of view. Let us suppose that in the standard quantum steering scenario, the trusted party observes a change in the state of its local system based on the input-output of the distant party. For any entangled state generated by the source, the trusted party would observe such an effect. However, if such a change for any input-output of the distant party can also be explained via some local hidden state model (LHS), then the state is deemed unsteerable in this scenario. Contrary to this, in the swap-steering scenario considered above, the trusted party detects whether its local state is entangled or not. Such an entanglement can never occur between two subsystems, if they can be individually described by an LHS model, regardless of any operation of the distant party. Consequently, there is a one-to-one correspondence between entanglement and network steerability. 

Our work raises several follow-up problems. The most important among them would be to extend the above swap-steering scenario to the multipartite regime and construct witnesses for every entangled state. The concept of swap-steering was recently generalised to ring networks in \cite{Baheti_2026}, which naturally motivates the exploration of its extension to other network topologies. Moreover, it would be interesting if every bipartite entangled state is swap-steerable even if the trusted party performs a single measurement. A more involving problem in this regard would be if one could remove the assumption of trust in the network and construct witnesses for every entangled state. 

{\it{Acknowledgments---}}
We thank Remiguisz Augusiak for discussions. This project was funded within the QuantERA II Programme (VERIqTAS project) that has received funding from the European Union’s Horizon 2020 research and innovation programme under Grant Agreement No 101017733 and the National Science Center, Poland, grant Opus 25, 2023/49/B/ST2/02468.

\input{ref.bbl}
\onecolumngrid
\appendix

\section{Witnesses}
\setcounter{fakt}{0}
\begin{fakt}
    Consider the swap-steering scenario depicted in Fig. 1 of the manuscript and the functional $\mathcal{S}_{\rho:NPT}$ given by
    \begin{eqnarray}\label{Witnpt}
    \mathcal{S}_{\rho:NPT}=\sum_{\substack{m,n=0\\m<n}}^{d-1}\alpha_m\alpha_n\Big[p((-,m,n),0)-p((+,m,n),0)\Big]-\sum_{m=0}^{d-1}\alpha_m^2p(m,0)
\end{eqnarray}
with Alice performing the measurement given above Eq. 8 of the manuscript.
    The maximal value attainable of $\mathcal{S}_{\rho:NPT}$ using an SOHS model is $\beta_{SOHS}=0$. 
\end{fakt}
\begin{proof}
    Let us recall from Eq. 2 of the manuscript that for correlations admitting a SOHS model, we have that
    \begin{equation}
        p(a,b)= \sum_{\lambda_1,\lambda_2} p(\lambda_1)p(\lambda_2)p(a|\ \rho_{\lambda_1}\otimes\rho_{\lambda_2})p(b|\lambda_1,\lambda_2)
    \end{equation}
for all $a,b$. Consequently, we have from \eqref{Witnpt} that
\begin{eqnarray}\label{A2}
     \mathcal{S}_{\rho:NPT}&=&\sum_{\lambda_1,\lambda_2} p(\lambda_1)p(\lambda_2) \Gamma(\rho_{\lambda_1}\otimes\rho_{\lambda_2})p(0|\lambda_1,\lambda_2)
\end{eqnarray}
where 
\begin{eqnarray}\label{gamma}
    \Gamma(\rho_{\lambda_1}\otimes\rho_{\lambda_2})=\sum_{\substack{m,n=0\\m<n}}^{d-1}\alpha_m\alpha_n\Big[p(-,m,n|\rho_{\lambda_1}\otimes\rho_{\lambda_2})-p(+,m,n|\rho_{\lambda_1}\otimes\rho_{\lambda_2})\Big]-\sum_{m=0}^{d-1}\alpha_m^2p(m|\rho_{\lambda_1}\otimes\rho_{\lambda_2}).
\end{eqnarray}
Expanding the above formula \eqref{gamma}, by considering the measurement of Alice $U_{A_1}^{\dagger}\mathcal{A}_{NPT}U_{A_1}$, we obtain that
\begin{eqnarray}\label{gamma2}
\Gamma(\rho_{\lambda_1}\otimes\rho_{\lambda_2})&=&\sum_{\substack{m,n=0\\m<n}}^{d-1}
\Tr\Big[\alpha_m\alpha_n(\proj{\delta_{m,n}^+}-\proj{\delta_{m,n}^+})\tilde{\rho}_{\lambda_1}\otimes\rho_{\lambda_2}\Big]-\sum_{m=0}^{d-1}\Tr\Big[\alpha_m^2\proj{\delta_m}\tilde{\rho}_{\lambda_1}\otimes\rho_{\lambda_2}\Big]\nonumber\\
&=& -\sum_{\substack{m,n=0\\m<n}}^{d-1}
\Tr\Big[\alpha_m\alpha_n(\ket{mn}\!\bra{nm}+\ket{nm}\!\bra{mn})\tilde{\rho}_{\lambda_1}\otimes\rho_{\lambda_2}\Big]-\sum_{m=0}^{d-1}\Tr\Big[\alpha_m^2\proj{mm}\tilde{\rho}_{\lambda_1}\otimes\rho_{\lambda_2}\Big]
\end{eqnarray}
where $\tilde{\rho}_{\lambda_1}=U_{A_1}\rho_{\lambda_1}U_{A_1}^{\dagger}$. Let us now consider $\ket{\tilde{\eta}}=\sum_{i}\alpha_i\ket{ii}$ and observe that
\begin{eqnarray}
    \proj{\tilde{\eta}}^{T_A}=\sum_{i,j=0}^{d^2-1}\alpha_i\alpha_j\ket{ji}\!\bra{ij}&=&\sum_{i=0}^{d-1}\alpha_i\alpha_j\proj{jj}+\sum_{\substack{i,j=0\\i<j}}^{d-1}\alpha_i\alpha_j(\ket{ji}\!\bra{ij}+\ket{ij}\!\bra{ji}).
\end{eqnarray}
Consequently, from \eqref{gamma2}, we have that 
\begin{eqnarray}
    \Gamma(\rho_{\lambda_1}\otimes\rho_{\lambda_2})=-\Tr\Big[\proj{\tilde{\eta}}^{T_A}\tilde{\rho}_{\lambda_1}\otimes\rho_{\lambda_2}\Big]=-\Tr\Big[\proj{\tilde{\eta}}\tilde{\rho}_{\lambda_1}^{T}\otimes\rho_{\lambda_2}\Big]\leq0.
\end{eqnarray}
Thus, we have from \eqref{A2} that for correlations admitting a SOHS model
\begin{eqnarray}
     \mathcal{S}_{\rho:NPT}\leq0.
\end{eqnarray}
This concludes the proof.
\end{proof}

\begin{fakt}
    Consider the swap-steering scenario depicted in Fig. 1 of the manuscript and the functional $\mathcal{S}_{\rho:CCN}$
    \begin{eqnarray}\label{Witccn}
    \mathcal{S}_{\rho:CCN}&=&\sum_{l_1,l_2=0}^{d-1}p(l_1l_2,l_1l_2)
\end{eqnarray}
with Alice performing the measurement given above Eq. 15 of the manuscript.
     The maximal value attainable of $\mathcal{S}_{\rho:CCN}$ using an SOHS model is $\beta_{SOHS}=\frac{1}{d}$. 
\end{fakt}
\begin{proof}
    The proof follows the same lines as presented in \cite{sarkar15, sarkar11}. Let us now consider the steering functional $\mathcal{S}_{\rho:CCN}$ in Eq. \eqref{Witccn} and express it in terms of the SOHS model [see Eq. 2 of the manuscript] as
\begin{eqnarray}
   \mathcal{S}_{\rho:CCN}= \sum_{l_1,l_2=0}^{d-1}\sum_{\lambda_1,\lambda_2} \ p(\lambda_1,\lambda_2)p(a|\rho_{\lambda_1}\otimes\rho_{\lambda_2})p(a|\lambda_1,\lambda_2)
    \leq \sum_{\lambda_1,\lambda_2} \ p(\lambda_1,\lambda_2)\max_{l_1l_2}\{p(a|\rho_{\lambda_1}\otimes\rho_{\lambda_2})\}
\end{eqnarray}
where we used the fact that $\sum_{l_1,l_2}p(l_1l_2|\lambda_1,\lambda_2)=1$ for any $\lambda_1,\lambda_2$. Maximising over $\rho_{\lambda_1},\rho_{\lambda_2}$ gives us
\begin{eqnarray}
     \sum_{\lambda_1,\lambda_2} \ p(\lambda_1,\lambda_2)\max_{a}\{p(a|\rho_{\lambda_1}\otimes\rho_{\lambda_2})\}\leq \sum_{\lambda_1,\lambda_2} \ p(\lambda_1,\lambda_2)\max_{\rho_{\lambda_1},\rho_{\lambda_2}}\max_{a}\{p(a|\rho_{\lambda_1}\otimes\rho_{\lambda_2})\}.
\end{eqnarray}
As $\sum_{\lambda_i}p(\lambda_i)=1$ for $i=1,2$ and the above function being linear 
allows us to conclude that 
\begin{eqnarray}
    \beta_{SOHS}\leq \max_{\ket{\psi}_{A_1},\ket{\psi}_{A_2}}\max_{l_1l_2}\{p({l_1l_2}|\ \ket{\psi}_{A_1},\ket{\psi}_{A_2})\}.
\end{eqnarray}
Now, considering the measurements of trusted Alice and the optimizing over pure states $\ket{\psi}_{A_1},\ket{\psi}_{A_2}\in \mathbbm{C}^d$ gives us $\beta_{SOHS}\leq\frac{1}{d}$. This bound can be saturated when the sources prepare the maximally mixed $U_{A_1}\rho_iU_{A_1}^{\dagger}=\frac{1}{d}\sum_{k=0}^{d-1}\proj{k}_{A_i}\proj{k}_{B_i}$ and the measurement with Bob is $M_B=\{\proj{k l}\}_{B_0B_1}$ for all $k,l=0,\ldots,d-1$. This concludes the proof.
\end{proof}

\begin{fakt}
    Consider the swap-steering scenario depicted in Fig. 1 of the manuscript and the functional $\mathcal{S}_{\tilde{\rho}}$ \begin{equation}\label{univwit}
\mathcal{S}_{\tilde{\rho}}=c_{0,0}p_B(0)+\sum_{a=0}^{d^2-1}c_{a,01}p(a,0|01)+\sum_{a,k=0}^{d^2-1}c_{a,1j}p(a,0|1k)
\end{equation}
with Alice performing the tomographically complete measurements suggested above Eq. 19 of the manuscript.
The maximal value attainable of $\mathcal{S}_{\tilde{\rho}}$ using an SOHS model is $\beta_{SOHS}=0$. 
\end{fakt}
\begin{proof}
    Let us recall from Eq. 2 of the manuscript that for correlations admitting a SOHS model, we have that
    \begin{equation}
        p(a,b)= \sum_{\lambda_1,\lambda_2} p(\lambda_1)p(\lambda_2)p(a|\ \rho_{\lambda_1}\otimes\rho_{\lambda_2})p(b|\lambda_1,\lambda_2)
    \end{equation}
for all $a,b$. Consequently, we have from \eqref{univwit} that
\begin{eqnarray}\label{A21}
     \mathcal{S}_{\tilde{\rho}}&=&\sum_{\lambda_1,\lambda_2} p(\lambda_1)p(\lambda_2) \Gamma(\rho_{\lambda_1}\otimes\rho_{\lambda_2})p(0|\lambda_1,\lambda_2)
\end{eqnarray}
where 
\begin{eqnarray}\label{gamma11}
    \Gamma(\rho_{\lambda_1}\otimes\rho_{\lambda_2})=c_{0,0}+\sum_{a=0}^{d^2-1}c_{a,01}p(a|\rho_{\lambda_1}\otimes\rho_{\lambda_2})+\sum_{a,k=0}^{d^2-1}c_{a,1j}p(a|1k,\rho_{\lambda_1}\otimes\rho_{\lambda_2}).
\end{eqnarray}
Expanding the above formula \eqref{gamma11}, by considering the observables of Alice $\{A_{00},A_{01},A_{11},\ldots,A_{1d-1}\}$, and recalling that $c_{0,0}=-\lambda_{0,0},c_{a,01}=-\sum_{j=1}^{d^2-1}\lambda_{0,j}\omega^{aj}_{d^2}$ and $ c_{a,1k}=-\sum_{j=1}^{d^2-1}\lambda_{k,kj\oplus d^2}\omega^{aj}_{d^2}$, we obtain
\begin{eqnarray}\label{gamma21}
\Gamma(\rho_{\lambda_1}\otimes\rho_{\lambda_2})&=&-\lambda_{0,0}-\sum_{a=0}^{d^2-1}\sum_{j=1}^{d^2-1}\lambda_{0,j}\omega^{aj}_{d^2}\Tr(\proj{\delta_{a,01}}\rho_{\lambda_1}\otimes\rho_{\lambda_2})-\sum_{a,k=0}^{d^2-1}\sum_{j=1}^{d^2-1}\lambda_{k,kj\oplus d^2}\omega^{aj}_{d^2}\Tr(\proj{\delta_{a,1k}}\rho_{\lambda_1}\otimes\rho_{\lambda_2})\nonumber\\
\end{eqnarray}
where $\ket{\delta_{a,ik}}$ are the eigenvectors of $A_{ik}$ for all $i,k$. Recalling that $A_{ik}^j=\sum_{a}\omega^{aj}_{d^2}\proj{\delta_{a,ik}}$, we obtain from the above expression
\begin{eqnarray}
    \Gamma(\rho_{\lambda_1}\otimes\rho_{\lambda_2})&=&-\lambda_{0,0}-\sum_{j=1}^{d^2-1}\lambda_{0,j}\Tr(A_{01}^j\rho_{\lambda_1}\otimes\rho_{\lambda_2})-\sum_{k=0}^{d^2-1}\sum_{j=1}^{d^2-1}\lambda_{k,kj\oplus d^2}\Tr(A_{1k}^j\rho_{\lambda_1}\otimes\rho_{\lambda_2}).
\end{eqnarray}
Let us now notice that putting the observables $A_{ik}$ and recalling that $W_{\tilde{\rho}}=\sum_{i,j=0}^{d^2-1}\lambda_{i,j}\omega_{d^2}^{ij(d^2-1)/2}X^i_{d^2}Z^j_{d^2}$, we obtain that
\begin{eqnarray}
    \Gamma(\rho_{\lambda_1}\otimes\rho_{\lambda_2})&=&-\Tr(W_{\tilde{\rho}}\rho_{\lambda_1}\otimes\rho_{\lambda_2})\leq 0.
\end{eqnarray}
Thus, we have from \eqref{A21} that for correlations admitting a SOHS model
\begin{eqnarray}
     \mathcal{S}_{\tilde{\rho}}\leq0.
\end{eqnarray}
This concludes the proof.
\end{proof}

\section{Necessity of entanglement in both sources}

Consider that the source $S_1$ generates a separable state $\rho_{A_1B_1}=\sum_kp_k\tau_k\otimes\tau'_k$. Let us now evaluate the unnormalised post-measured state of Alice given any measurement of Bob $\{M_b\}$ as
\begin{eqnarray}
\rho^{PM}_{b,A_1A_2}=\sum_kp_k\Tr_{B_1B_2}(\I_A\otimes M_b(\tau_k)_{A_1}\otimes(\tau'_k)_{B_1}\otimes \rho_{A_2B_2}).
\end{eqnarray}
This implies that the state $\rho^{PM}_{A_1A_2}$ is separable as
\begin{eqnarray}\label{ApB1}
\rho^{PM}_{b,A_1A_2}=\sum_kp_k(\tau_k)_{A_1}\otimes\Tr_{B_1B_2}(\I_A\otimes M_b(\tau'_k)_{B_1}\otimes \rho_{A_2B_2})=\sum_kp_k(\tau_k)_{A_1}\otimes(\tau''_{k,b})_{A_2}.
\end{eqnarray}

Let us now evaluate the witness $\mathcal{S}_{\tilde{\rho}}$ \eqref{univwit} in this case. For this purpose, we observe that $p(a,0|01)$ and $p(a,0|1k)$ in \eqref{univwit} can be expressed using \eqref{ApB1} as
\begin{eqnarray}
p(a,0|x)=\sum_kp_k\Tr(\tau_k\otimes\tau''_{k,0}M_{a,x})
\end{eqnarray}
where $x\equiv01,1k$ and $M_{a,x}$ are the corresponding measurement elements. This allows us to express \eqref{univwit} as
\begin{equation}
\mathcal{S}_{\tilde{\rho}}=-\sum_kp_k\sum_{i,j=0}^{d^2-1}\lambda_{i,j}\omega_{d^2}^{ij(d^2-1)/2}\Tr(X^i_{d^2}Z^j_{d^2}\tau_k\otimes\tau''_{k,0})
\end{equation}
which implies that
\begin{eqnarray}
   \mathcal{S}_{\tilde{\rho}}= -\sum_kp_k\Tr(W_{\tilde{\rho}}\tau_k\otimes\tau''_{k,0})\leq0.
\end{eqnarray}
The other proposed witnesses can be evaluated the same way to obtain similar conclusion when one of the source generates a separable state. Consequently, to violate the above proposed witnesses both sources should generate entangled states.

\end{document}

%% file: ref.bbl
\providecommand{\noopsort}[1]{}\providecommand{\singleletter}[1]{#1}%